\documentclass[conference]{IEEEtran}
\IEEEoverridecommandlockouts

\usepackage{float}
\usepackage{cite}
\usepackage{amsmath,amssymb,amsfonts}
\usepackage{amsthm}
\usepackage{algorithm}
\usepackage{algpseudocode}
\usepackage{graphicx}
\usepackage{textcomp}
\usepackage{xcolor}
\usepackage{stfloats}
\usepackage{subfigure}

\usepackage[margin=0.75in]{geometry}

\newtheorem{theorem}{Theorem}
\newtheorem{lemma}{Lemma}

\newtheorem{assumption}{Assumption}
\newtheorem{corollary}{Corollary}

\def\BibTeX{{\rm B\kern-.05em{\sc i\kern-.025em b}\kern-.08em
    T\kern-.1667em\lower.7ex\hbox{E}\kern-.125emX}}

\begin{document}

\newgeometry{top=1in, bottom=0.75in, left=0.75in, right=0.75in}

\title{Resource Allocation with Multi-Team Collaboration Based on Hamilton's Rule\\
    \thanks{
        $^\dag$R. Karam, R. Lin, B. Butler, and M. Egerstedt are with the Samueli School of Engineering, University of California, Irvine, Irvine, CA, 92697, USA. Email: {\tt\small \{rwkaram, rlin10, bbutler2, magnus\}@uci.edu}. This research was supported in part by an appointment to the Intelligence Community Postdoctoral Research Fellowship Program at the University of California, Irvine administered by Oak Ridge Institute for Science and Education (ORISE) through an interagency agreement between the U.S. Department of Energy and the Office of the Director of National Intelligence (ODNI). This work is also sponsored in part by the Air Force Research Laboratory, Munitions Directorate (RWTA), Eglin AFB, FL (Award No. FA8651-24-2-0001). Opinions, findings and conclusions, or recommendations are those of the authors and do not necessarily reflect the views of the sponsoring agencies.
    }
}

\author{\IEEEauthorblockN{Riwa Karam, Ruoyu Lin, Brooks A. Butler, and Magnus Egerstedt}}

\maketitle

\begin{abstract}
    This paper presents a multi-team collaboration strategy based on Hamilton's rule from ecology that facilitates resource allocation among multiple teams, where agents are considered as shared resource among all teams that must be allocated appropriately. We construct an algorithmic framework that allows teams to make bids for agents that consider the costs and benefits of transferring agents while also considering relative mission importance for each team. This framework is applied to a multi-team coverage control mission to demonstrate its effectiveness. It is shown that the necessary criteria of a mission evaluation function are met by framing it as a function of the locational coverage cost of each team with respect to agent gain and loss, and these results are illustrated through simulations. 
\end{abstract}

\section{Introduction}

Multi-robot systems are increasingly deployed in large-scale and safety-critical applications such as search-and-rescue, environmental monitoring, and disaster response, e.g., \cite{lim2009market,dasgupta2011effects,kim2024area,butler2024collaborative,butler2024resilience,lin2023dynamic}. Some of these tasks often require teams of robots to operate in parallel across geographically distinct regions, while balancing between mission priorities for individual teams and overall system performance.
A key challenge in such scenarios is the optimal allocation of available robots across multiple teams, weighed by the priorities of the individual team missions, driving the entire system to an optimized overall effectiveness.

This challenge connects directly to the broader literature on resource allocation in multi-agent systems, e.g., \cite{anussornnitisarn2005decentralized,jiang2022collaborative, banerjee2016multi}, with applications including cyber-physical systems \cite{afrin2021resource}, UAV team deployment \cite{dai2022multi}, and smart grid scheduling \cite{nair2018multi}. For instance, \cite{dai2022multi} combines reinforcement learning with communication resource allocation, for UAV deployment, optimizing coverage and communication. In addition to computation resource distribution, other works consider collaboration in manufacturing \cite{li2018multi} and the collaborative management of supply chains \cite{fu2015adaptive}. Game-theoretic frameworks have also been proposed to handle dynamically changing environments, where task demands evolve over time and robot allocations must adapt accordingly, e.g., \cite{park2021multi}. In practice, however, assigning an additional robot to a team can directly increase that team’s mission effectiveness.

A common inspiration for modeling systems is the environment and nature surrounding us; understanding how it moves and works has been an influence to previous literature work, as in \cite{bonabeau1999swarm, benyus1997biomimicry, bar2005biomimetics, vincent2006biomimetics}. In this paper, we are inspired by and extract the properties of Hamilton’s rule \cite{hamilton1964genetical}. This rule, originally formulated in evolutionary biology, states that altruistic behavior emerges when the weighted benefit to others exceeds the cost to the donor. By analogy, we consider how teams of robots may relate to each other in terms of relative mission importance, and we view the robots themselves as transferable resources that may be shared among cooperating teams.

Building on this perspective, we design a framework in which teams compare both the relative importance of their missions and the marginal cost or benefit of reallocating robots. Transfers are permitted only when they improve the performance of the overall multi-team system. We demonstrate this idea in the context of a Voronoi-based coverage control problem, e.g., \cite{cortes2004coverage}, showing how altruistic decisions lead to improved system outcomes. Unlike prior work \cite{lin2023predator} that focused on a single team of cooperating robots, our setting explicitly addresses multiple teams of robots, each initially assigned to a distinct region but capable of reallocating robots for a collective benefit.

The remainder of this paper is outlined as follows. In Section~\ref{sec:ham_rule_framework}, we present a collaborative framework inspired by Hamilton's rule that facilitates the transfer of agents between teams depending on the relative costs of agent transfer and mission importance for each team. We then show how this framework can be applied in a coverage control setting in Section~\ref{sec:coverage}, demonstrate its performance through simulation in Section~\ref{sec:simulations}, and provide concluding remarks in Section~\ref{sec:conclusion}.

\section{A Team-Level Collaboration Framework Based on Hamilton's Rule} \label{sec:ham_rule_framework}

\parfillskip=0pt
In this section, we formulate the problem of resource allocation in multi-agent teams based on Hamilton's rule. To begin with, we introduce a graph-based model, e.g., \cite{mesbahi2010graph}, to represent team interactions. Let \( G = (V, E) \) be the undirected graph abstracting the \( m \) different available teams of agents as well as their potential interactions, where \( V = \{v_1, v_2, \dots, v_m\} \) is the set of nodes representing the teams, and \( E \subseteq V \times V \) is the set of edges representing the 
\restoregeometry \noindent \setlength{\parfillskip}{0pt plus 1fil}
potential collaborative interactions between them.
\begin{assumption}
    \label{ass:homogeneous}
    All agents in all \( m \) teams are homogeneous; i.e., all agents are identical.
\end{assumption}
In this paper, we compute potential \textbf{collaborations} between different teams of \textbf{homogeneous} agents (given Assumption~\ref{ass:homogeneous}) using Hamilton's rule. Each team $k$ is assigned a mission evaluation function $F_k(n_k)$, denoting the team's performance at achieving its assigned mission, where $n_k$ denotes the number of agents allocated to team $k$. The mission evaluation function is assumed to be strictly increasing, to emphasize that having a higher value in $F_k(n_k)$ is equivalent to team $k$ performing better at its assigned mission, and having diminishing returns in the improvement of team $k$'s performance -hence, highlighting the fact that after a certain point, allocating additional resources to team $k$ produces progressively smaller gains in performance. This property incurs a more balanced distribution of resources among teams. Those assumptions are given by the following:
\begin{assumption}
    \label{ass:mission_eval_func}
    For each team $k$, the mission evaluation function \( F_k(n_k) : \mathbb{N} \to \mathbb{R} \)
    satisfies the following properties:
    \begin{enumerate}
        \item \textbf{Strictly Increasing:} For all $n \geq 1$,
        \begin{equation} \nonumber
            F_k(n+1) > F_k(n).
        \end{equation}
        \item 
        \textbf{Discrete Concavity:} For all $n \geq 2$,
        \begin{equation} \nonumber
            F_k(n + 1) - F_k(n) \leq F_k(n) - F_k(n - 1).
        \end{equation}
    \end{enumerate}
\end{assumption}
It is important to note that an increasing number of agents in a team will increase that team's performance with diminishing returns; however, these returns will eventually saturate and start exhibiting decreasing returns \cite{ferber1999multi}. In this paper, we consider areas of interest where an overcrowding of agents is not possible, and where the number of agents in each team at any point in time will not reach a point where decreasing returns occur.

In ecology, Hamilton’s rule was originally formulated to explain the evolution of altruistic behavior among genetically related individuals \cite{hamilton1964genetical}. According to Hamilton’s rule, an altruistic act is favored by natural selection if the \textit{benefit to the recipient}, weighted by the \textit{genetic relatedness} between the donor and recipient, exceeds the \textit{cost to the donor}. This is expressed as \( rB \geq C \), where \( r \) is the coefficient of relatedness, \( B \) is the benefit to the recipient, and \( C \) is the cost to the donor. In our multi-team scenario, Hamilton's rule is used to decide whether the exchange of an agent between pairwise teams is beneficial. Here, the role of genetic relatedness is replaced by a weight ratio that reflects the relative importance or ``mission priority" of the teams, such that
\begin{equation} \nonumber
    r_{ij} = \frac{w_j}{w_i} \quad \text{and} \quad r_{ji} = \frac{w_i}{w_j} = \frac{1}{r_{ij}},
\end{equation}
where $w_k > 0$ is a positive weight representing the mission's importance. For this work, we adapt Hamilton's rule to this problem setting as:
\begin{equation} \label{eq:HR}
    r_{ij}B_j > C_i,
\end{equation}
where $B_j$ represents the benefit for team $j$ upon receiving an additional agent, and $C_i$ represents the cost for team $i$ of losing that agent. Note that, in our case, we consider when collaboration is strictly beneficial and exclude scenarios where it is equally beneficial.

\subsection{Pairwise Uni-Directionality of Hamilton's Rule} \label{sub:uni_dir}

Suppose we have two teams of homogeneous agents, each executing a mission in separate regions of interest. For our problem formulation, we define, for a team \( k \in V \),
\begin{align}
    B_k = F_k(n_k + \delta) \, - \, F_k(n_k), \label{eq:B_k} \\
    C_k = F_k(n_k) \, - \, F_k(n_k - \delta), \label{eq:C_k}
\end{align}
where \( \delta \in \mathbb{N} \) is the number of agents to be exchanged between two neighboring teams. In our work, we consider the exchange of a single robot per each update, i.e., \( \delta = 1\).
Thus, the first condition from \eqref{eq:HR} that must be satisfied for team \( i \) to give an agent to team \( j \), given that \( (i, j)\in E \), becomes
\begin{equation} \label{eq:HR_new}
    r_{ij} \Bigl[F_j(n_j + 1) - F_j(n_j)\Bigr] \, > \, \Bigl[F_i(n_i) - F_i(n_i - 1)\Bigr],
\end{equation}
where $F_k(n_k + 1)$ denotes the value of team $k$'s evaluation function after receiving an agent, and $F_k(n_k - 1)$ denotes it after losing an agent. Since $F_k$ is strictly increasing, we have:
\begin{equation} \nonumber
    F_k(n_k + 1) \, > \, F_k(n_k) \, > \, F_k(n_k - 1).
\end{equation}
This condition ensures that the weighted marginal gain for team $j$ exceeds the marginal loss for team $i$; diminishing returns imply that the benefit of an additional agent is greater, as well as the cost of losing an agent is higher, when a team has fewer agents.

For the sake of our work, a collaboration between two teams of homogeneous agents is beneficial when Hamilton's rule, as defined in equation~\eqref{eq:HR_new}, is satisfied, and is the migration of one agent from one team to the other. However, when considering pairwise teams, it would be redundant to have Hamilton's rule be satisfied in both directions, i.e., it being beneficial for team \( i \) to give team \( j \) an agent and for team \( j \) to give team \( i \) an agent. Hence, proving the \textbf{uni-directionality} property of Hamilton's rule is essential.
\begin{theorem}\label{thm:uni_dir}
Consider two teams $i$ and $j$, with \( i,j \in V \), each with at least one agent.
Under Assumptions~\ref{ass:homogeneous} and~\ref{ass:mission_eval_func}, if \( r_{ij} \, B_j > C_i \), it follows that \( r_{ji} \, B_i \leq C_j \).
\end{theorem}
\begin{proof}
    We prove the result by contradiction. Assume that both $r_{ij} \, B_j > C_i$ and $r_{ji} \, B_i > C_j$ hold simultaneously. Writing these two inequalities in terms of weights, we get $w_j \, B_j > w_i \, C_i $ and $ w_i \, B_i > w_j \, C_j $, leading to 
    \begin{equation*}
    w_i\,B_i > \left(w_i\,\frac{C_i}{B_j}\right) C_j.
    \end{equation*}
    Canceling $w_i$ (with $ w_i > 0 $) yields
    \begin{equation} \label{eq:wrong_ineq}
        B_i \, B_j > C_i \, C_j.
    \end{equation}
    However, $F_k$ having diminishing returns as a property, given in Assumption~\ref{ass:mission_eval_func}, implies that for each team $k$,
    \begin{equation} \nonumber
        B_k = F_k(n_k + 1) - F_k(n_k) \leq F_k(n_k) - F_k(n_k - 1) = C_k,
    \end{equation}
    by~\eqref{eq:B_k} and~\eqref{eq:C_k}. Thus, we have $B_i \leq C_i$ and $B_j \leq C_j$, which can be rewritten as $C_i = B_i + \alpha$ and $C_j = B_j + \beta$, such that $\alpha, \beta \geq 0$. Then, 
    \begin{equation*}
    C_i \, C_j = (B_i + \alpha)(B_j + \beta) = B_i \, B_j + \beta B_i + \alpha B_j + \alpha\beta.
    \end{equation*}
    Since $\beta B_i$, $\alpha B_j$, and $\alpha\beta$ are all non-negative, we have 
    \begin{equation*}
    B_i \, B_j \leq C_i \, C_j,
    \end{equation*}
    which contradicts the inequality~\eqref{eq:wrong_ineq}.
    Thus, since all agents are identical, i.e., have equal ``resource contribution", under Assumption~\ref{ass:homogeneous}, the assumption that both inequalities hold must be false, proving Theorem~\ref{thm:uni_dir}.
\end{proof}

\subsection{Multi-Team Bidding Process} \label{sub:bid_proc}

\begin{figure*}[!t]
    \centering
    \includegraphics[width=2\columnwidth]{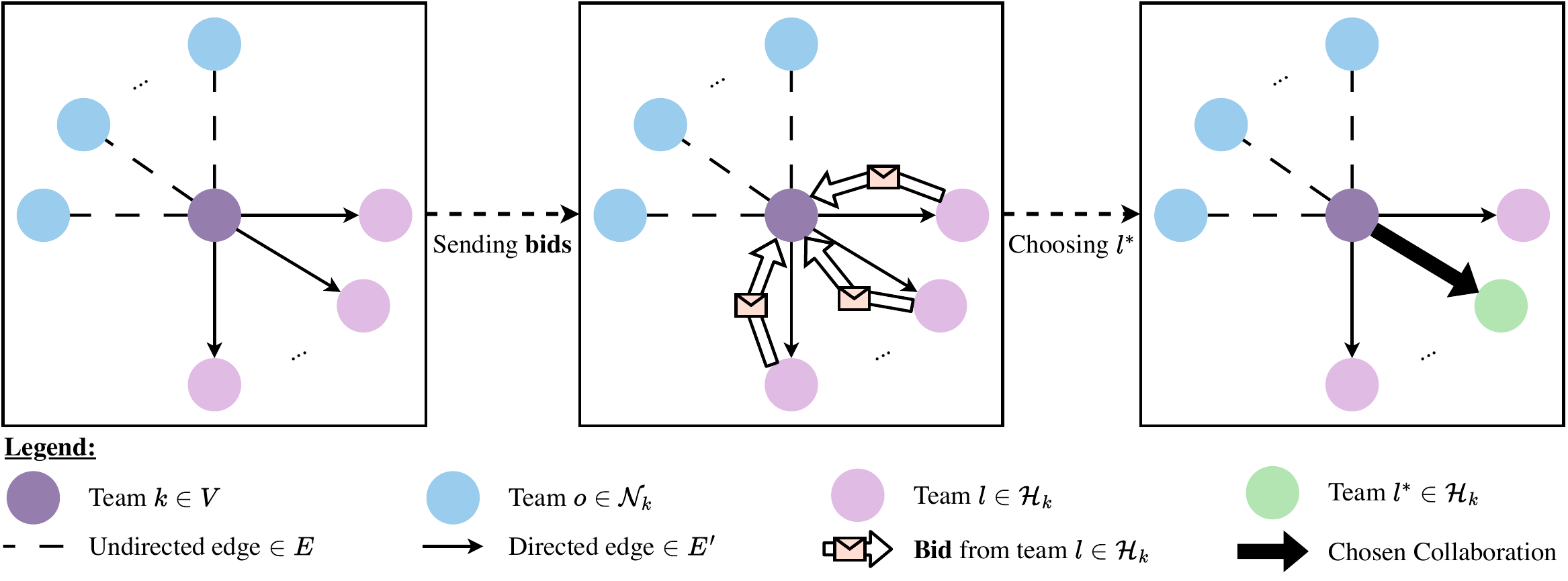}
    \caption{Bidding process for team \( k \in V \) to choose team \( l^* \in \mathcal{H}_k \) between all \( l \in \mathcal{H}_k \). This bidding process happens according to Equations~\eqref{eq:bid_out} and~\eqref{eq:opt_bid_out}, and makes sure that each team having multiple outgoing collaborations, will choose the collaboration with the maximum net gain benefit.}
    \label{fig:bid_collab_choice}
    \vspace{-12pt}
\end{figure*}

In Section~\ref{sub:uni_dir}, we focused on pairwise collaborations using Hamilton's Rule. However, when considering a system with multiple teams, the problem is inherently more challenging. Instead of only evaluating local, pairwise interactions, the overall system performance must be taken into account. We say that given two teams \( i,j \in E\), an outgoing collaboration from team \( i \) to team \( j \) and an incoming collaboration to team \( i \) from team \( j \) is beneficial when Equation~\ref{eq:HR} holds. We must note that if a team ends up with \( 0 \) robots after collaborating, in the next collaboration process, that team will only be considered for incoming collaborations.

Having proved that pairwise collaborations are unidirectional in Section~\ref{sub:uni_dir}, we must now consider how each team will handle its potential collaborations with its neighbors. Let \( \mathcal{N}_k \subseteq E \) be the set of team $k$'s neighbors. Let \( \mathcal{H}_k \subseteq \mathcal{N}_k \) be the set of team $k$'s neighbors with which team $k$ has an outgoing collaboration. Team $k$ will choose to only collaborate with one of its neighbors depending on the following \textbf{bid}: each neighbor $l \in \mathcal{H}_k$ will send its expected performance improvement (i.e., net gain) to team $k$ such that
\begin{equation}
    \begin{aligned}
        \Delta_{k \to l} &= r_{kl}B_l - C_j \\
                 &= w_l[F_l(n_l + 1) - F_l] \\
                 &\quad- w_k[F_k(n_k) - F_k(n_k - 1)].
    \end{aligned}
    \label{eq:bid_out}
\end{equation}
Team $k$ will then collaborate with $l^\mathbb{*}$ such that
\begin{equation}
    l^\mathbb{*} = \arg \max_{l\in \mathcal{H}_k} \Delta_{k \to l}. \label{eq:opt_bid_out}
\end{equation}
This bidding process is shown in Figure~\ref{fig:bid_collab_choice}.
Let \( \mathcal{I}_k \subseteq \mathcal{N}_k \) be the set of team $k$'s neighbors with which team $k$ has an incoming collaboration with. Then, in a similar manner to the previous bid for outgoing collaborations, team $k$ will choose to only receive an agent from one of its neighbors depending on the following \textbf{bid}: each neighbor $e \in \mathcal{I}_k$ will send team $k$'s expected performance improvement to team $k$, such that
\begin{equation}
    \begin{aligned}
        \Delta_{e \to k} &= r_{ek}B_k - C_e \\
                 &= w_k[F_k(n_k + 1) - F_k] \\
                 &\quad - w_e[F_e(n_e) - F_e(n_e - 1)].
    \end{aligned}
    \label{eq:bid_in}
\end{equation}
Team $k$ will then collaborate with $e^\mathbb{*}$ such that
\begin{equation}
    e^\mathbb{*} = \arg \max_{e\in \mathcal{I}_k} \Delta_{e \to k}. \label{eq:opt_bid_in}
\end{equation}

We introduce a \textit{global mission evaluation function} that aggregates the performance of all teams. For \( m \) teams, the global objective is defined as:
\begin{equation} \label{eq:global_func}
    \mathcal{G}(n_1, n_2, \dots, n_m) = \sum_{k=1}^{m} w_k \, F_k(n_k).
\end{equation}
This global function encapsulates the collective performance of the entire system, and is used to determine whether a proposed collaboration, or set of collaborations, improves the overall system effectiveness. Equations~\eqref{eq:bid_out},~\eqref{eq:opt_bid_out},~\eqref{eq:bid_in}, and~\eqref{eq:opt_bid_in} guarantee a maximum local increase in \( \mathcal{G} \). In a multi-team setting, decisions must consider not only individual pairwise benefits, but also the net effect on the global objective \( \mathcal{G} \). As such, the bidding and collaboration protocols must be extended to evaluate and coordinate transfers among more than two teams simultaneously. In the following, we discuss a potential approach for implementing these multi-team collaboration strategies.

\begin{figure*}[!t]
    \centering
    \includegraphics[width=2\columnwidth]{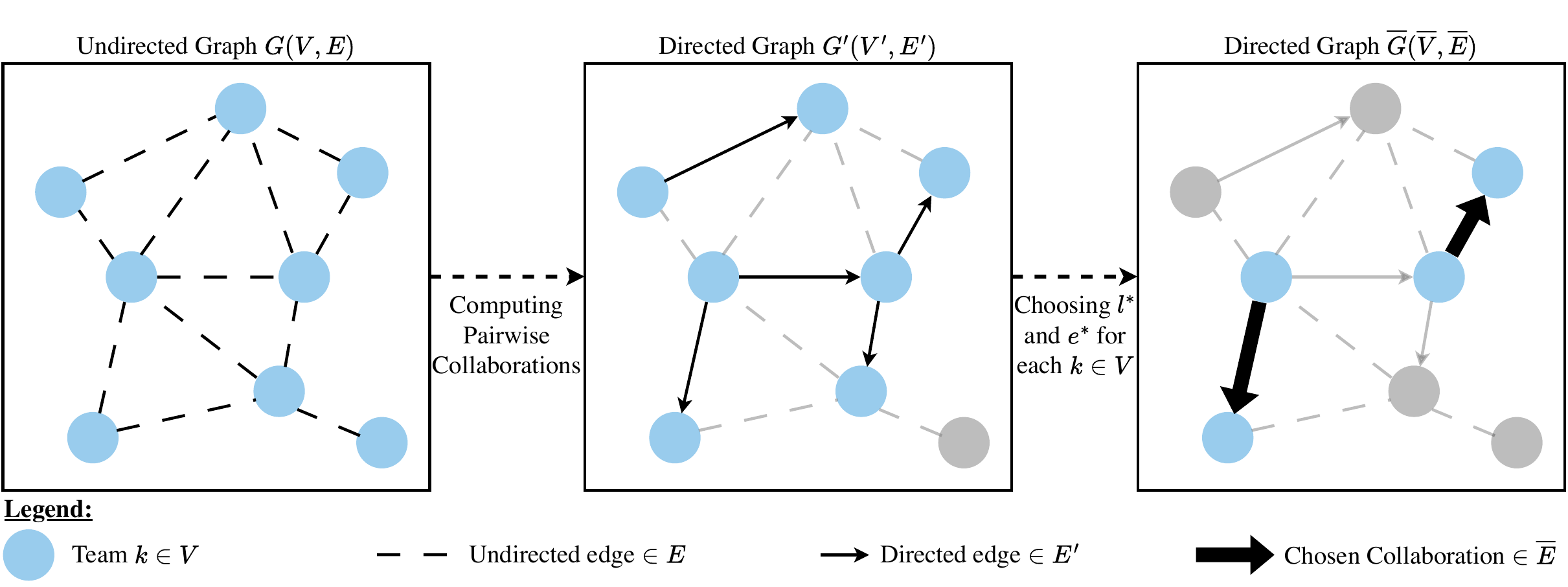}
    \caption{Graphs filtering process: from an original undirected input graph, to its filtered directed graph of collaborations satisfying Equation~\eqref{eq:HR}, and then to its filtered directed graph of chosen pairwise collaborations using the bidding process introduced in Section~\ref{sub:bid_proc}.}
    \label{fig:graph_filtering}
    \vspace{-12pt}
\end{figure*}

Given the initial \emph{undirected} graph $G(V, E)$, computing the allowed collaborations and defining their directions using Hamilton's rule as defined in Equation~\eqref{eq:HR_new}, leads to a \emph{directed} graph of filtered edges and nodes. Let this directed graph be defined as \( G'(V', E')\) such that \( V' \subseteq V \) and \( E' \subseteq E \), with \( V' \) being the set of teams that \emph{can} collaborate with at least one of its neighbors, and \( E' \) being the set of allowed collaborations. \( G' \) is then a filtered graph, representing the pairwise, locally beneficial collaborations between teams, ruled by Hamilton's rule. \( \overline{G}(\overline{V}, \overline{E}) \) will be the directed graph stemming from \( G' \), with \( \overline{E} \) containing only the edges representing the outgoing collaborations between each team \( k \) and its corresponding team \( l^* \) and the incoming collaborations between each team \( k \) and its corresponding team \( e^* \), and \( \overline{V} \) containing only the teams corresponding to the edges in \( \overline{E} \).
To simplify our analysis, we make the following assumption on the number of agents for the entire system.
\begin{assumption} \label{ass:num_agents}
    The number of agents available across all \( m \) teams in the system is fixed and equal to $N$.
\end{assumption}
In other words, we assume that during the collaborative process, no new agents are added to the system and no agents are lost (i.e., the total number of agents in the system is conserved).
To determine whether or not to allow this set of collaborations \( \overline{E} \) to be executed, the overall global performance should increase if executed. Hence, the system will be driven by the following inequality rule:
\begin{equation} \label{eq:global_rule}
    \mathcal{G}(n_1', \dots, n_m') \, > \mathcal{G}(n_1, \dots, n_m),
\end{equation}
where \( n_k' \) is the updated number of agents in team $k$ if the collaborations in the set \( E' \) were to be executed, under Assumption~\ref{ass:num_agents}. This rule allows for a global benefit instead of a greedy, pairwise benefit in collaborating; in a more linguistic expression, this rule dictates that ``if all of us benefit, each of us benefits."

\subsection{Iterative Collaboration Framework}

In this section, we present the framework for multi-team collaboration as an algorithm. This framework starts with an undirected graph \( G(V, E) \), filters it to a directed graph \( G'(V', E') \) using Hamilton’s Rule in~\eqref{eq:HR_new}, and then has each team select its best collaboration partner using the two bidding processes introduced in Section~\ref{sub:bid_proc}, and eventually get to \( \overline{G}(\overline{V}, \overline{E}) \). If executing these migrations increases the global mission evaluation function, they are performed, as shown in~\eqref{eq:global_rule}. This framework, shown in Algorithm~\ref{alg:collab_alg}, is repeated until no further improvement is possible. Figure~\ref{fig:graph_filtering} clarifies, through an example, how the graph filtering process works.

Let the set
\begin{equation} \nonumber
    \Omega \;=\; \Bigl\{(n_1,\dots,n_m) \in \mathbb{Z}_{>0}^m 
    \;\Big|\; \sum_{k=1}^m n_k = N \Bigr\}
\end{equation}
be the set of all possible allocations of the $N$ agents between the $m$ teams. The global mission evaluation function introduced in~\eqref{eq:global_rule} will be \( \mathcal{G} : \Omega \;\to\; \mathbb{R} \).

\begin{algorithm}[!b]
    \caption{Multi-Team Collaboration Process}
    \label{alg:collab_alg}
    \begin{algorithmic}[1]
        \State \textbf{Input:} Initial allocation $(n_1, \dots, n_m)$ satisfying $\sum_{k=1}^m n_k = N$, mission evaluation functions $F_k$, weights $w_k > 0$.
        \State Compute the initial global objective: \( \mathcal{G}^{(0)} \)
        \State Set iteration counter $t \gets 0$.
        \Repeat
            \State Construct $G(V,E)$
            \State Using Eq.~\eqref{eq:HR_new}, get the set of allowed collaborations
            \State Construct $G'(V',E')$
            \For{each team $k \in V'$}
                \For{each $l \in \mathcal{H}_k$}
                    \State Compute the bid \( \Delta_{k \to l} \) in Eq.~\eqref{eq:bid_out}
                \EndFor
                \State Select the neighbor with the highest bid (Eq.~\eqref{eq:opt_bid_out})
            \EndFor
            \For{each team $k \in V'$}
                \For{each $e \in \mathcal{I}_k$}
                    \State Compute the bid \( \Delta_{e \to k} \) in Eq.~\eqref{eq:bid_in}
                \EndFor
                \State Select the neighbor with the highest bid (Eq.~\eqref{eq:opt_bid_in})
            \EndFor
            \State Construct $\overline{G}(\overline{V}, \overline{E})$
            \State Compute the new global objective: \( \mathcal{G}^{(t+1)} \)
            \If{$\mathcal{G}^{(t+1)} > \mathcal{G}^{(t)}$}
                \For{each team $k \in \overline{V}$}
                    \State \emph{Collaborate}: transfer one agent from $k$ to $l^*$
                \EndFor
                \State Update the allocation $(n_1, \dots, n_m)$
            \EndIf
            \State Set $t \gets t+1$.
        \Until{$\mathcal{G}^{(t)} \ge \mathcal{G}^{(t + 1)}$}
        \State \textbf{Output:} The final allocation $(n_1^*, \dots, n_m^*)$.
    \end{algorithmic}
\end{algorithm}

\begin{lemma} \label{lem:existence_optimal}
    Under Assumption~\ref{ass:num_agents}, there exists an allocation
    \begin{equation} \nonumber
        (n_1^*, \dots, n_m^*) \;\in\; \Omega
    \end{equation}
    such that
    \begin{equation} \nonumber
        \begin{aligned}
            \mathcal{G}(n_1^*, \dots, n_m^*) 
        \; \; \ge \; \; 
        \mathcal{G}(n_1, \dots, n_m), \quad
        \forall \, (n_1, \dots, n_m) \in \Omega.
        \end{aligned}
    \end{equation}
    In other words, $\mathcal{G}$ attains a maximum value on $\Omega$, and $(n_1^*,\dots,n_m^*)$ is an optimal allocation.
\end{lemma}
\begin{proof}
    Since $n_k > 0$ and $\sum_{k=1}^m n_k = N$ (Assumption~\ref{ass:num_agents}), the set $\Omega$ is finite and non-empty. Since $\Omega$ is finite, the set $\{f(n) \mid n \in \Omega\}$ is also finite, and any finite subset of $\mathbb{R}$ has a maximum. Hence, there exists at least one $n^\star \in \Omega$ such that 
    \begin{equation} \nonumber
        f(n^\star) = \max\{ f(n) \mid n \in \Omega \}.
    \end{equation}
    Thus, $\mathcal{G}$ has at least one global maximizer $(n_1^*,\dots,n_m^*) \in \Omega$, proving the lemma.
\end{proof}

\begin{theorem} \label{thm:alg_convergence}
    Under Assumptions~\ref{ass:homogeneous}-\ref{ass:num_agents}, the iterative process described in Algorithm~\ref{alg:collab_alg} terminates in a finite number of steps and converges to an allocation $(n_1^*, \dots, n_m^*) \in \Omega$ that maximizes $\mathcal{G}$.
\end{theorem}
\begin{proof}
    Let $\{x^{(t)}\}_{t=0}^\infty$ denote the sequence of allocations generated by the algorithm, where 
    \begin{equation*}
    x^{(t)} = (n_1^{(t)}, n_2^{(t)}, \dots, n_m^{(t)}) \in \Omega. 
    \end{equation*}
    At each iteration $t$, if at least one collaboration (transfer) is executed, then by the design of Algorithm~\ref{alg:collab_alg}, we have 
    \begin{equation*}
    \mathcal{G}\bigl(x^{(t+1)}\bigr) > \mathcal{G}\bigl(x^{(t)}\bigr). 
    \end{equation*}
    Since $\Omega$ is finite and $\mathcal{G}$ is a real-valued function on $\Omega$, by Lemma~\ref{lem:existence_optimal} the function $\mathcal{G}$ attains a maximum value on $\Omega$, denoted by
    \begin{equation*}
    M = \max_{x \in \Omega} \mathcal{G}(x). 
    \end{equation*}
    Thus, the strictly increasing sequence $\{\mathcal{G}(x^{(t)})\}$ is bounded above by $M$. Because the sequence takes values in a finite set, it cannot increase indefinitely and must eventually stabilize; hence, no further collaborations would occur from that point on. More formally, there exists some $T \in \mathbb{N}$ such that 
    \begin{equation*}
    \mathcal{G}\bigl(x^{(T+1)}\bigr) \leq \mathcal{G}\bigl(x^{(T)}\bigr), 
    \end{equation*}
    where no further collaborations occur at $T+1$.
    Since a collaboration is executed only when it yields a strict increase in $\mathcal{G}$ (see Line 22 of Algorithm~\ref{alg:collab_alg}), we have that at iteration $T$ no allowed transfer can produce a strict increase. Thus, $x^{(T)}$ is a local optimum with respect to the allowed moves.
    
    Next, we show that under Assumption~\ref{ass:mission_eval_func},
    any local optimum with respect to these transfers is in fact a global optimum. Suppose for contradiction that $x^{(T)}$ is not a global optimum. Then there exists some allocation $x^* \in \Omega$ such that
    \begin{equation*}
    \mathcal{G}(x^*) > \mathcal{G}(x^{(T)}). 
    \end{equation*}
    Since the global objective $\mathcal{G}$ is point-wise concave in each coordinate (owing to the point-wise concavity of each $F_k$ by Assumption~\ref{ass:mission_eval_func} and the fact that $w_k > 0$), the function $\mathcal{G}$ has the property that any local improvement, i.e., any feasible agent transfer from one team to another that is allowed by Hamilton's Rule, would yield a strictly higher value of $\mathcal{G}$. Therefore, the existence of an allocation $x^*$ with $\mathcal{G}(x^*) > \mathcal{G}(x^{(T)})$ implies that there exists at least one allowed transfer from $x^{(T)}$ that would increase $\mathcal{G}$, contradicting the termination condition of the algorithm.
    Thus, the terminal allocation $x^{(T)}$ must coincide with a global maximizer of $\mathcal{G}$ on $\Omega$. Since the process only permits transfers that strictly increase $\mathcal{G}$ and $\Omega$ is finite, the iterative process converges to the optimal allocation in a finite number of steps, thus proving our theorem.
\end{proof}
Given the finite-time convergence of Algorithm~\ref{alg:collab_alg} to an optimal global mission evaluation, we now apply this framework to scenarios where area coverage may be used as a metric for local team mission evaluation.

\section{Team Mission: Area Coverage} \label{sec:coverage}

To illustrate the operation of the proposed framework, we consider the mission assigned to each robot team to be optimal area coverage, to which a typical approach is the Voronoi-based coverage control introduced in \cite{cortes2004coverage}.
Specifically, we consider the deployment of a team of $N \in \mathbb{Z}^+$ robots to cover a domain of interest $\mathcal{D} \subset \mathbb{R}^d$, which is a measurable set with nonempty interior, i.e., $\mathrm{Int}(\mathcal{D}) \neq \emptyset$, associated with a measurable density function $\phi(q): \mathcal{D} \rightarrow \mathbb{R}^+$ encoding the relative importance of different points in $\mathcal{D}$. By assuming the farther (based on the Euclidean norm) a robot is from a point $q \in \mathcal{D}$, the less effectively it covers $q$, the locational cost function can be formulated as \cite{cortes2004coverage} 
\begin{equation}
    L(p) = \sum_{i \in \mathcal{N}} \int_{\mathcal{V}_i}\! \|p_i - q\|^2 \phi(q) \,dq,
    \label{eq:lcost}
\end{equation}
where $p_i \in \mathbb{R}^d$ denotes the position of robot $i$, the vector $p = [p_1^\top, p_2^\top,\dots, p_N^\top]^\top$ contains the positions of all robots, $\mathcal{N} := \{1, 2, \dots. N\}$ is the set of indices of the robots, and
\begin{align} 
    \mathcal{V}_i &= \left\{ {q}\in \mathcal{D} 
    \,|\, \lVert{q}- p_i\rVert \leq \lVert {q}-p_j\rVert,\,\, \forall j \neq i \in \mathcal{N} \right\} \label{voronoi}
\end{align}
is the Voronoi cell of robot $i$. One can notice from~\eqref{voronoi} that $\cup_{i=1}^N \mathcal{V}_i = \mathcal{D}$ and the Lebesgue measure $\lambda(\cap_{i=1}^N \mathcal{V}_i) = 0$.

A typical approach to find an optimal coverage configuration of the robots, e.g., \cite{cortes2004coverage}, is that robot $i$ follows against the gradient of $L(p)$ with respect to $p_i$, $\forall i \in \mathcal{N}$, which is known as a continuous-time version of Lloyd's algorithm, asymptotically leading to a centroidal Voronoi tessellation (CVT), i.e., 
\begin{align}
p_i = c_i = \frac{\int_{\mathcal{V}_i}\! {q} \phi({q}) \,d{q}}{\int_{\mathcal{V}_i}\!\phi({q}) \,d{q}}, \quad \forall i \in \mathcal{N}, \notag
\end{align}
where $c_i \in \mathbb{R}^d$ is the centroid of the Voronoi cell of robot $i$. Similar to $p$, we denote a CVT as $c = [c_1^\top, c_2^\top,\dots, c_N^\top]^\top$.

In this paper, we consider the total number of robots $N$ as an additional argument of the locational cost function \eqref{eq:lcost}, i.e., 
\begin{equation}
    L(N, p) = \sum_{i \in \mathcal{N}} \int_{\mathcal{V}_i}\! \|p_i - q\|^2 \phi(q)\,dq.
    \label{eq:lcost_new}
\end{equation}
In addition, Fig.\ref{fig:Voronoi} illustrates an example of how the Voronoi tessellation changes after adding one robot to a team of robots.

To meet Assumption~\ref{ass:mission_eval_func}, we let the mission evaluation function $F(N) = -L(N,c)$ for some $c$. In Theorem~\ref{thm:mcc} and Corollary~\ref{cor:cvt_cost}, we investigate the strictly decreasing property of $L(N,c)$ with respect to $N$, i.e., the strictly increasing property of $F(N)$ presented in Assumption~\ref{ass:mission_eval_func}.


\begin{figure}[!t]
    \centering
    \includegraphics[width=0.35\textwidth]{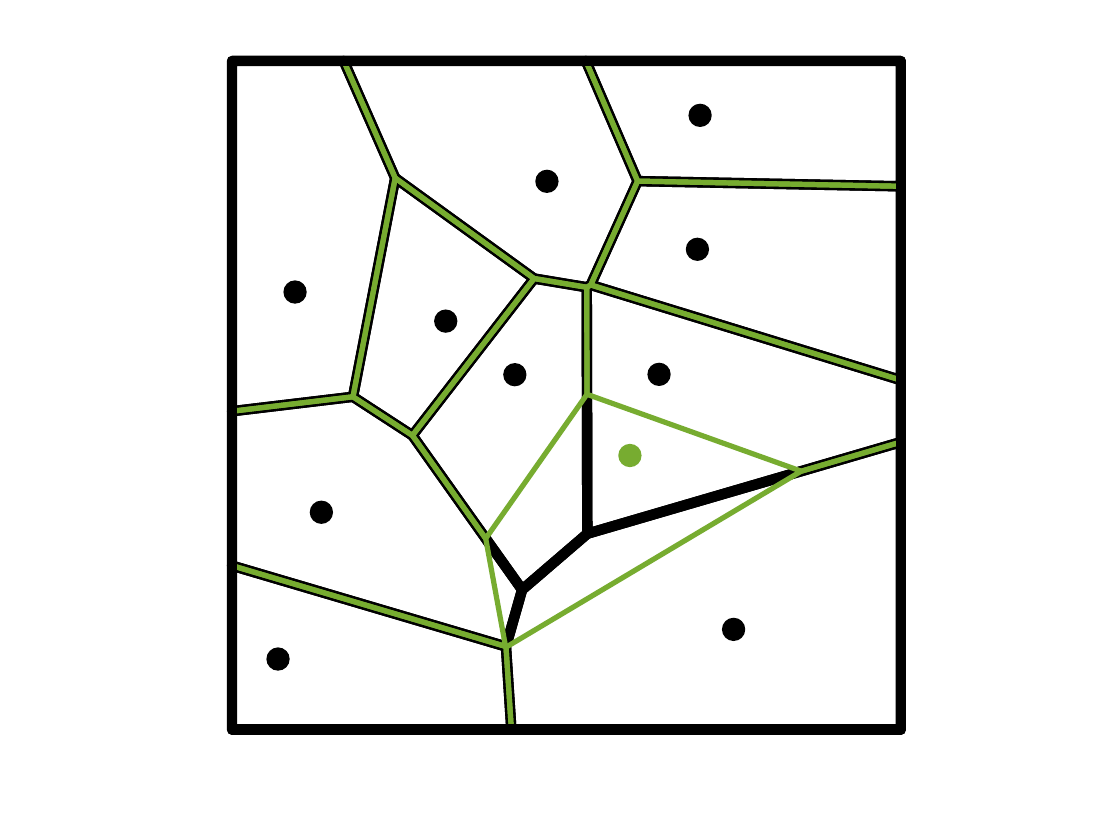}
    \caption{The boundaries of the Voronoi cells of $ N = 10 $ robots (black dots) in $\mathcal{D} \subset \mathbb{R}^2$ are shown by the black dashed lines. The newly added robot (green dot) represents $p_{N + 1}$, and the resulting Voronoi cells are represented by the green lines.}
    \label{fig:Voronoi}
    \vspace{-12pt}
\end{figure}

\begin{theorem}
    Suppose $p_i \in \mathcal{D}$, $\forall i \in \mathcal{N}$. By adding one more robot $p_{N+1} \in {\mathrm{Int}(\mathcal{D})} \setminus \{p_i\}$, $\forall i \in \mathcal{N}$, the locational cost~\eqref{eq:lcost_new} decreases, i.e., 
    \begin{equation}
        L(N+1, \, p^\prime)<L(N, \, p), \quad \forall N \in \mathbb{Z}^+,\notag
    \end{equation}
    where $p^\prime = [p^\top,p_{N+1}^\top]^\top$.
    \label{thm:mcc}
\end{theorem}
\begin{proof}
    Based on \cite{cortes2004coverage}, the locational cost function~\eqref{eq:lcost_new} can be rewritten as
    \begin{equation}
        L(N,p) = \int_{\mathcal{D}} \min_{i \in \mathcal{N}} \|p_i - q\|^2 \phi(q)\,dq. \notag
    \end{equation}
    In addition, we denote
    \begin{equation*}
    H(p^\prime,q) = \min_{i \in \mathcal{N}} \|p_i - q\|^2
    - 
    \min_{i \in \mathcal{N}^\prime} \|p_i - q\|^2,
    \end{equation*}
    where $\mathcal{N}^\prime = \mathcal{N} \cup \{N+1\}$. Apparently, $H(p^\prime,q) \geq 0$, $\forall q \in \mathcal{D}$.

Moreover, we let 
\begin{equation*}
\bar{\delta} = \min_{i\in\mathcal N}\|p_{N+1}-p_i\|>0.
\end{equation*}
Since $p_{N+1}\in \mathrm{Int}(\mathcal D)$, then $\exists \,{\varepsilon}>0$ such that the open Euclidean ball
\begin{equation*}
\mathcal{B}(p_{N+1},{\varepsilon}) = \{q \in \mathbb{R}^d \mid \|q - p_{N+1}\| < \varepsilon\} 
\subset
\mathcal{D}.
\end{equation*}
Furthermore, we define
\begin{equation*}
\bar{\varepsilon}=\min\{\varepsilon,\bar{\delta}/2\}
\quad\text{and}\quad
\bar{\mathcal{B}}=\mathcal{B}(p_{N+1},\bar{\varepsilon}).
\end{equation*}
Then, the Lebesgue measure $\lambda(\bar{\mathcal{B}})>0$, and thus we have
\begin{align*}
\|q-p_i\|
\geq 
&\,\|p_i-p_{N+1}\|-\|q-p_{N+1}\| \\
>
&\, 
\bar{\delta}-\bar{\varepsilon} \\
\geq &\,
\bar{\delta}/2
>
\|q-p_{N+1}\|,
\quad \forall q\in \bar{\mathcal{B}}, \, \forall i\in\mathcal N.
\end{align*}
Hence, we get
\begin{equation*}
\min_{i\in\mathcal N'}\|p_i-q\|^2=\|q-p_{N+1}\|^2<\min_{i\in\mathcal N}\|p_i-q\|^2,
\quad \forall q\in \bar{\mathcal{B}},
\end{equation*}
which means that $H(p^\prime,q) > 0$, $\forall q \in \bar{\mathcal{B}}$.
Together with the assumption $\phi(q) >0$ and the fact $H(p^\prime,q) \geq 0$, $\forall q \in \mathcal{D}$, we can have
\begin{align*}
&L(N,p)-L(N+1,p') \\
= &\int_\mathcal{D}\!
H(p^\prime,q)
\phi(q)
\,dq \\
= &\int_{\bar{\mathcal{B}}}\!
H(p^\prime,q)
\phi(q)
\,dq 
+
\int_{\mathcal{D} \setminus \bar{\mathcal{B}}}\!
H(p^\prime,q)
\phi(q)
\,dq \\
\geq&
\int_{\bar{\mathcal{B}}}\!
H(p^\prime,q)
\phi(q)
\,dq
> 
0,
\end{align*}
which proves $L(N+1,p')<L(N,p)$.
\end{proof}

\begin{figure}[!b]
    \vspace{-12pt}
    \centering
    \includegraphics[width=0.32\textwidth]{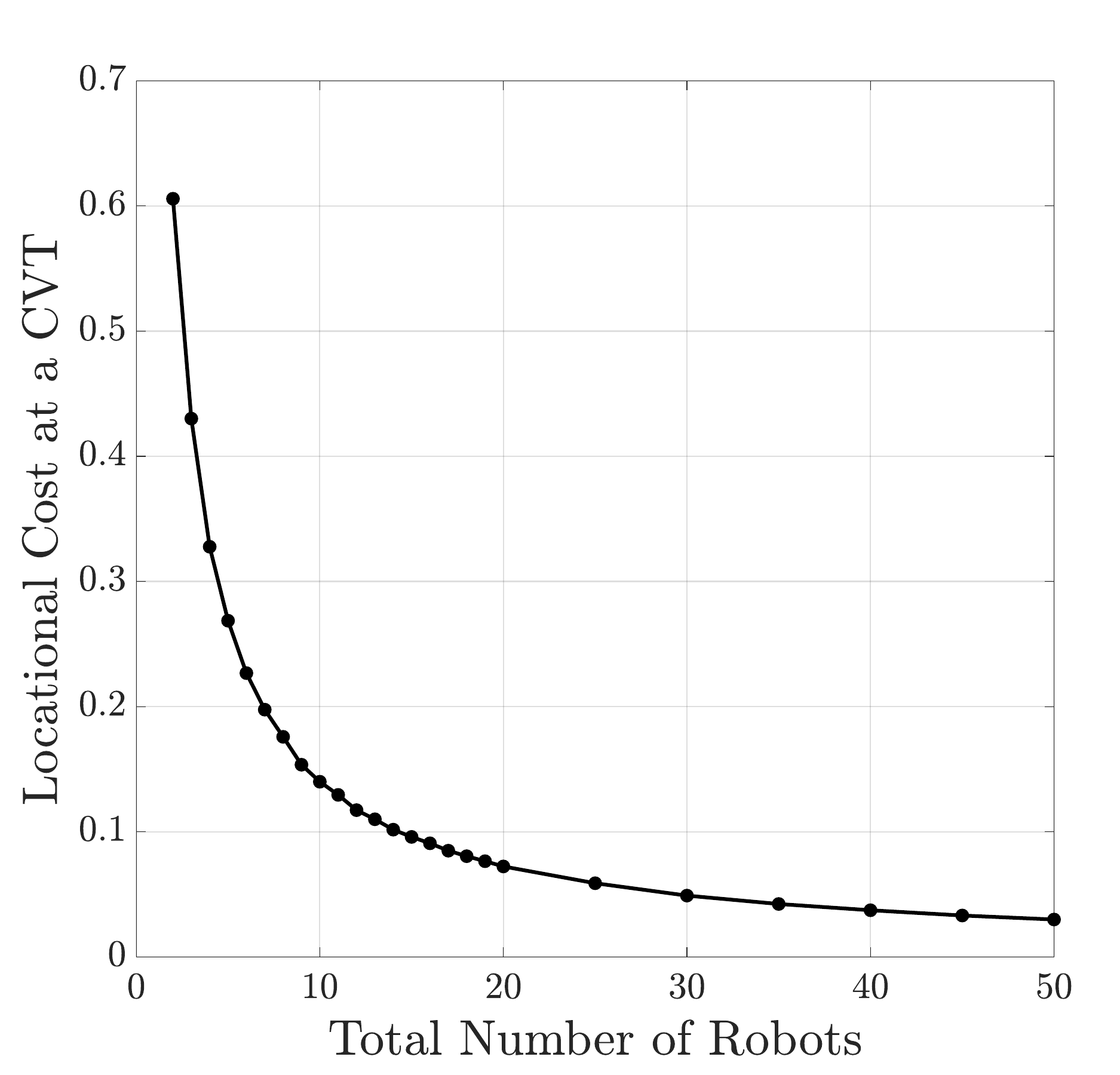}
    \caption{The locational cost \eqref{eq:lcost_new} at a CVT with $\phi(q) = e^{-(\frac{q_x^2}{0.8^2} + \frac{q_y^2}{0.8^2})}$ with respect to the total number of robots $N$.}
\label{fig:cost_cvt}
\end{figure}

\begin{corollary}
Suppose $N$ robots have converged to a CVT, i.e., $p_i = c_i$, $\forall i \in \mathcal{N}$. After adding one more robot $p_{N+1} \in {\mathrm{Int}(\mathcal{D})} \setminus \{p_i\}$, $\forall i \in \mathcal{N}$, a new CVT $c^\prime$ with $N+1$ robots resulting from the Lloyd's algorithm has a lower locational cost value than the old CVT $c$ with $N$ robots does, i.e.,
    \begin{equation*}
        L(N+1,c^\prime)<L(N,c), \quad \forall N \in \mathbb{Z}^+.
    \end{equation*}
\label{cor:cvt_cost}
\end{corollary}
\vspace{-20pt}
\begin{proof}
Suppose the $N\in \mathbb{Z}^+$ robots have converged to a CVT $c$, by adding $p_{N+1} \in {\mathrm{Int}(\mathcal{D})} \setminus \{p_i\}$, $\forall i \in \mathcal{N}$, as per Theorem~\ref{thm:mcc}, we have $L(N,c) > L(N+1, [c^T, p^T_{N+1}]^T)$. Then, the Lloyd's algorithm, e.g., \cite{cortes2004coverage}, results in
    \begin{equation*}
    L(N+1, [c^T, p^T_{N+1}]^T) \geq L(N+1, c^\prime),
    \end{equation*}
    where $c^\prime$ is a new CVT with $N+1$ robots. Hence, $L(N+1,c^\prime)<L(N,c), \,\, \forall N \in \mathbb{Z}^+$,
    proving Corollary~\ref{cor:cvt_cost}.
\end{proof}

\begin{figure*}[!b]
    \vspace{-12pt}
    \subfigure[]{\includegraphics[width=0.35\textwidth]{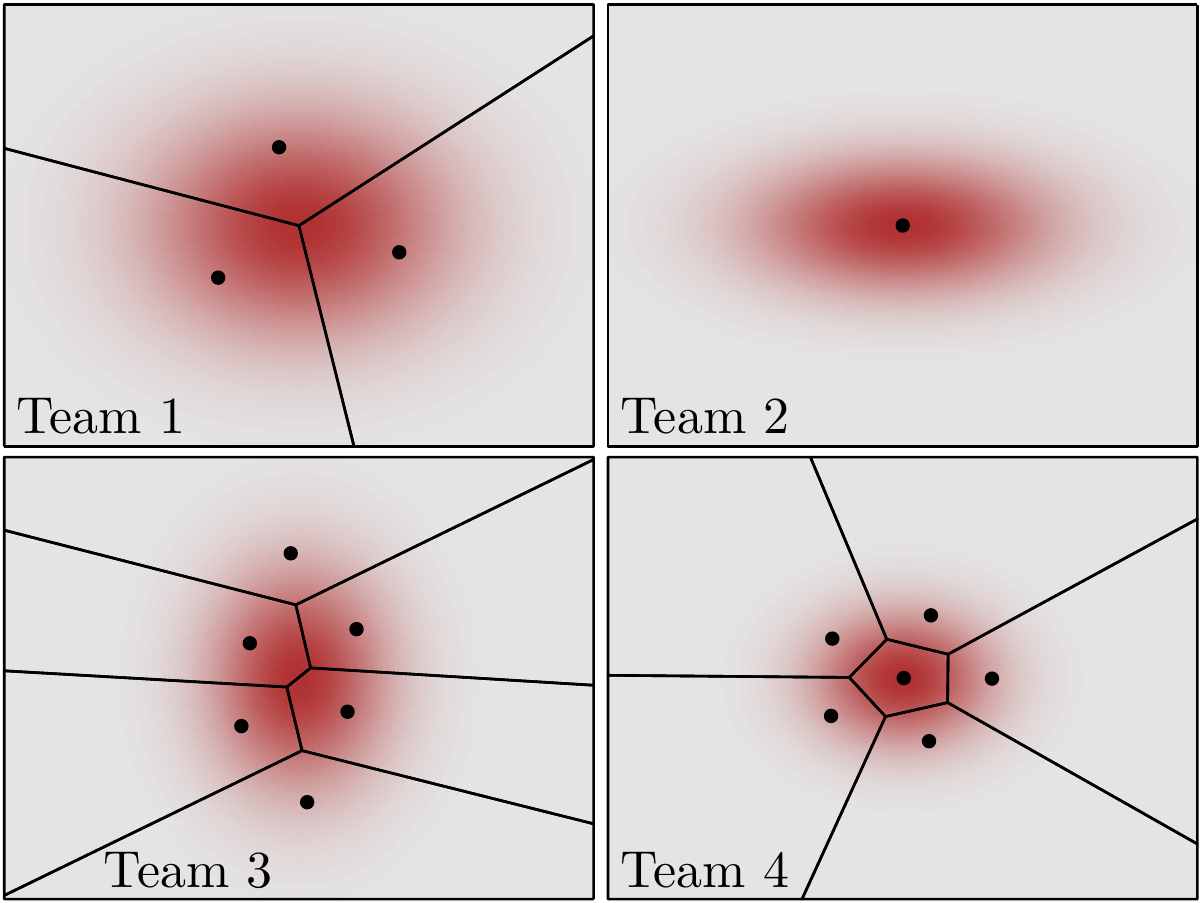}}
    \hfill
    \subfigure[]{\includegraphics[width=0.35\textwidth]{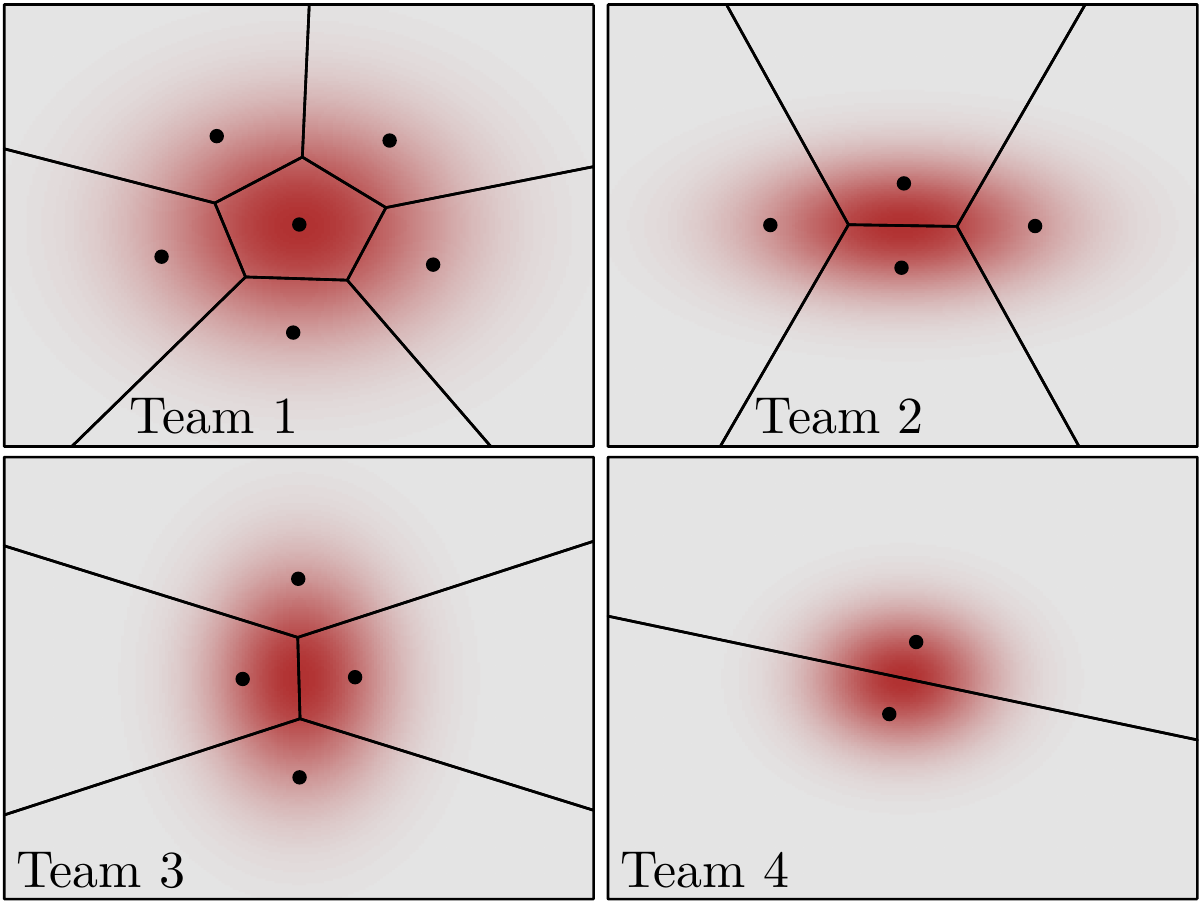}}
    \hfill
    \subfigure[]{\includegraphics[width=0.28\textwidth]{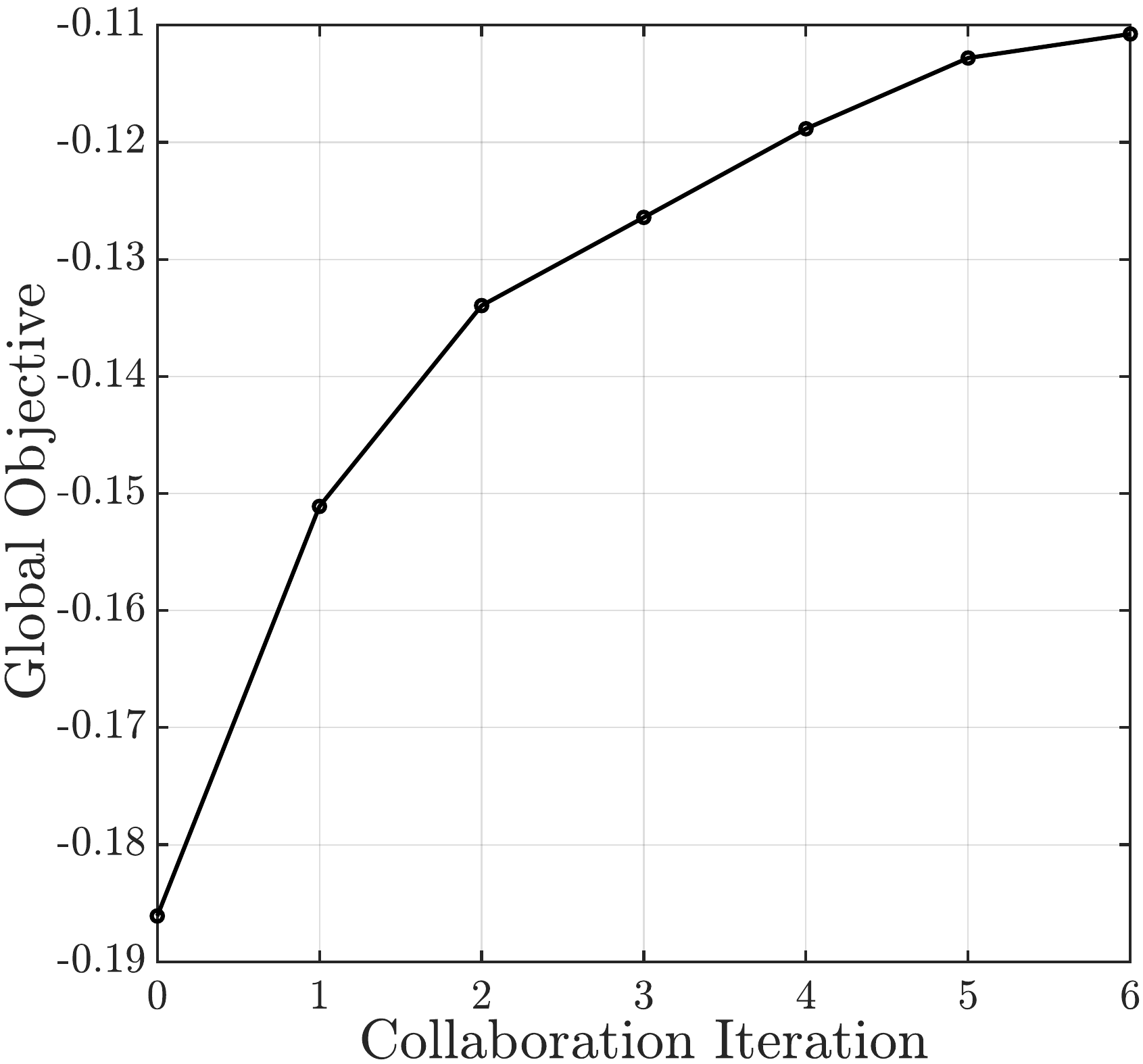}}
    \caption{Simulation of Algorithm~\ref{alg:collab_alg} with \( m = 4\) and \( N = 16 \) while having all teams have equal weights, but different mission evaluation functions by changing the density function \( \phi(q) \) in the locational cost in Equation~\ref{eq:lcost}. (a) The initial random allocation of robots to teams. (b) The final robot allocation. (c) The value of \( \mathcal{G} \) at each collaboration iteration until collaboration is no longer possible (hence, getting the allocation in (b)).}
    \label{fig:sim_same_w}
\end{figure*}

The property of diminishing returns of the mission evaluation function $F(N) = -L(N,c)$ is not trivial to prove in the context of Voronoi-based coverage control. Nevertheless, we empirically find that $L(N,c)$ asymptotically decreases with respect to $N$ with a diminishing rate, as illustrated in Fig.\ref{fig:cost_cvt}, which corresponds to the property of diminishing returns of $F(N)$ presented in Assumption~\ref{ass:mission_eval_func}. We leave its theoretical proof as a future work. 


\section{Simulations} \label{sec:simulations}

In this section, we show two different simulations of Algorithm~\ref{alg:collab_alg} applied to coverage control, where \( m = 4 \) and \( N = 16 \). Figures~\ref{fig:sim_same_w} and~\ref{fig:sim_same_phi} show the evolution of the four teams, each having been assigned a mission importance and mission evaluation function. As introduced in Section~\ref{sec:coverage}, coverage control fits our assumptions for team missions, as the negative of the locational cost satisfies the conditions in Assumption~\ref{ass:mission_eval_func}. Figure~\ref{fig:sim_same_w} emphasizes how having different team weights but identical mission evaluation functions across teams, can affect team-level collaboration, and final robot allocation, while Figure~\ref{fig:sim_same_phi} emphasizes how having different mission evaluation functions but same weights, will drive the collaborations and allocations.

In the simulation shown in Figure~\ref{fig:sim_same_w}, teams will find it beneficial to collaborate with each other mainly based on their mission evaluation function, hence, the difference in the density functions of their respective area, since the weights of all teams are equal and of value \( w = 1 \). The density functions of the respective areas assigned to Teams 1 to 4 are as follows: \( \phi_1(q) = e^{-(\frac{q_x^2}{0.5^2} + \frac{q_y^2}{0.5^2})} \), \( \phi_2(q) = e^{-(\frac{q_x^2}{0.5^2} + \frac{q_y^2}{0.3^2})} \), \( \phi_3(q) = e^{-(\frac{q_x^2}{0.3^2} + \frac{q_y^2}{0.5^2})} \), and \( \phi_4(q) = e^{-(\frac{q_x^2}{0.3^2} + \frac{q_y^2}{0.3^2})} \). Team 1 is assigned an area with a larger density function than the other teams, thus ending up with the most robots and vice-versa for Team 4. However, Teams 2 and 3 get allocated 4 robots each given that their densities have identical area spanning but different spanning axes. However, in the simulation shown in Figure~\ref{fig:sim_same_phi}, given that Teams 1 to 4 have gradually increasing weights, as follows: \( w_1 = 1 \), \( w_2 = 2 \), \( w_3 = 6 \), and \( w_4 = 20 \), but identical density functions of value \( \phi(q) = e^{-(\frac{q_x^2}{0.5^2} + \frac{q_y^2}{0.5^2})} \), at the end of collaboration, each end up with a respective gradually increasing number of robots. In both simulations, as seen in Figures~\ref{fig:sim_same_w}-(b, c) and~\ref{fig:sim_same_phi}-(b, c), the system stabilizes once no globally beneficial collaborations can be executed between teams.

To compute the cost \( C_i \) (removing a robot from team \( i \)) and benefit \( B_j \) (adding a robot to team \( j \)), we run Lloyd's algorithm until a CVT is reached. This is done due to the unfairness of computing the benefits and costs after the instantaneous removal or addition of a robot from or to a team, where the benefit will always be higher than the cost by several magnitudes. In other words, the workload assigned to a robot about to be removed is significantly higher than the workload that will be assigned to a robot about to be added because of the differences in the area each is responsible for. Thus, when removing a robot, a CVT would have been reached right before the point of transfer, inducing a significant performance gap. However, when adding a robot, it is being added to a team where the CVT is reached without it, creating a suboptimal positioning with respect to the density function \( \phi(q) \) and remaining agents. Hence, running Lloyd's algorithm to reach CVT would give a fair comparison between benefit and cost.

In Figures~\ref{fig:sim_same_w}-(c) and~\ref{fig:sim_same_phi}-(c), we plot the value of the global mission evaluation function, introduced in~\eqref{eq:global_func}. It is clear that not only is \( \mathcal{G} \) increasing, hence the system performance is improving, but its curve appears to exhibit diminishing returns in the improvements. Those experiments serve as a proof that coverage control is a suitable application to our framework, having the negative of the locational cost of area coverage be the mission evaluation function of the teams in the system. Moreover, the simulations exhibit behavior that fits the expected results according to the Assumptions provided for the proper execution of Algorithm~\ref{alg:collab_alg}.

 \begin{figure*}[!t]
    \subfigure[]{\includegraphics[width=0.35\textwidth]{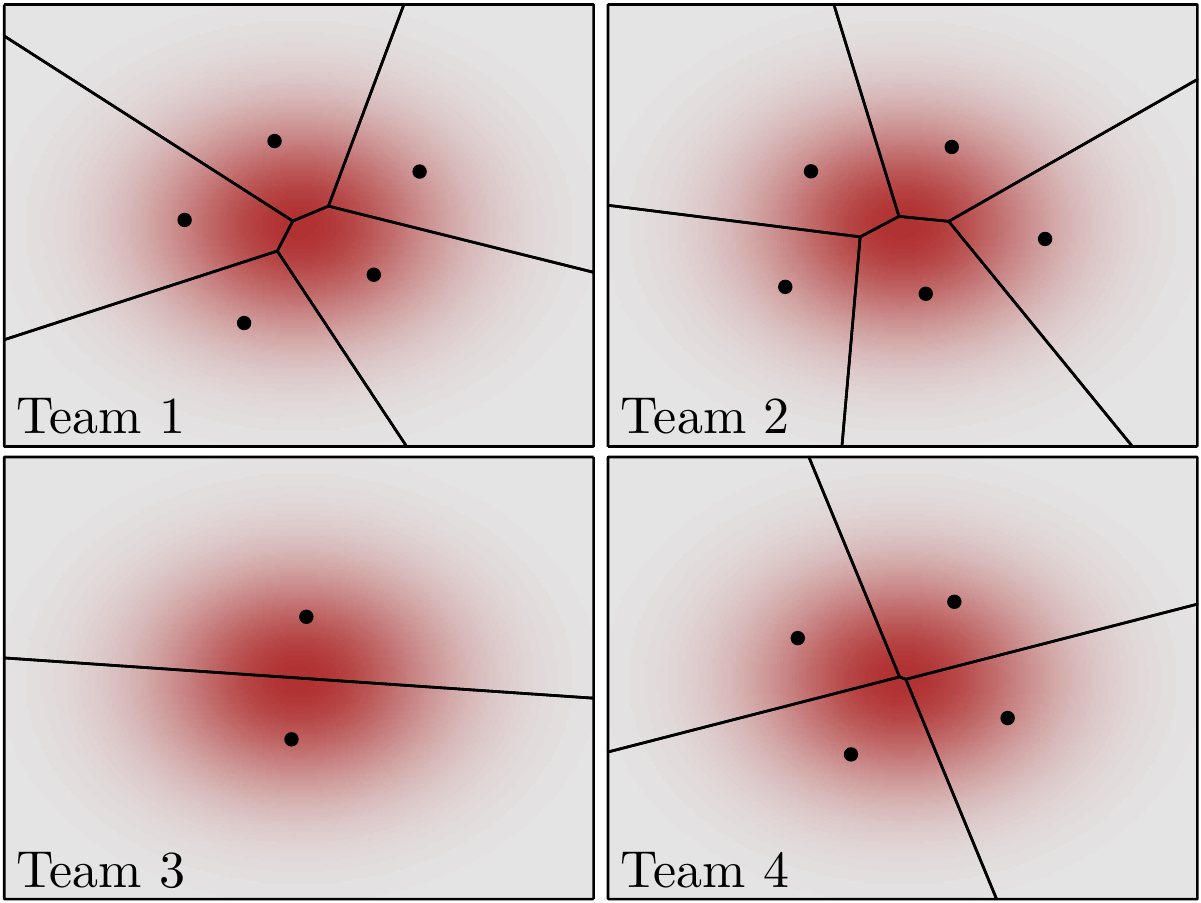}}
    \hfill
    \subfigure[]{\includegraphics[width=0.35\textwidth]{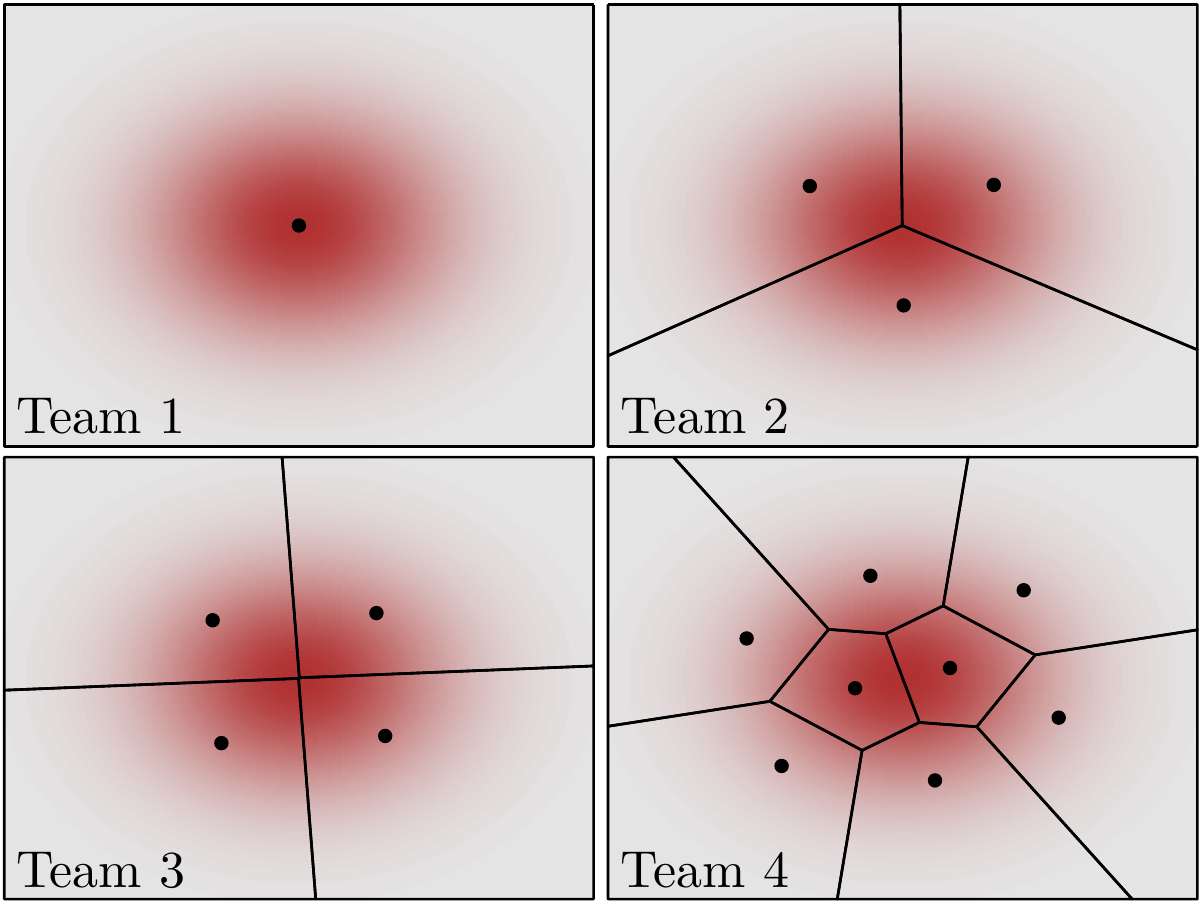}}
    \hfill
    \subfigure[]{\includegraphics[width=0.28\textwidth]{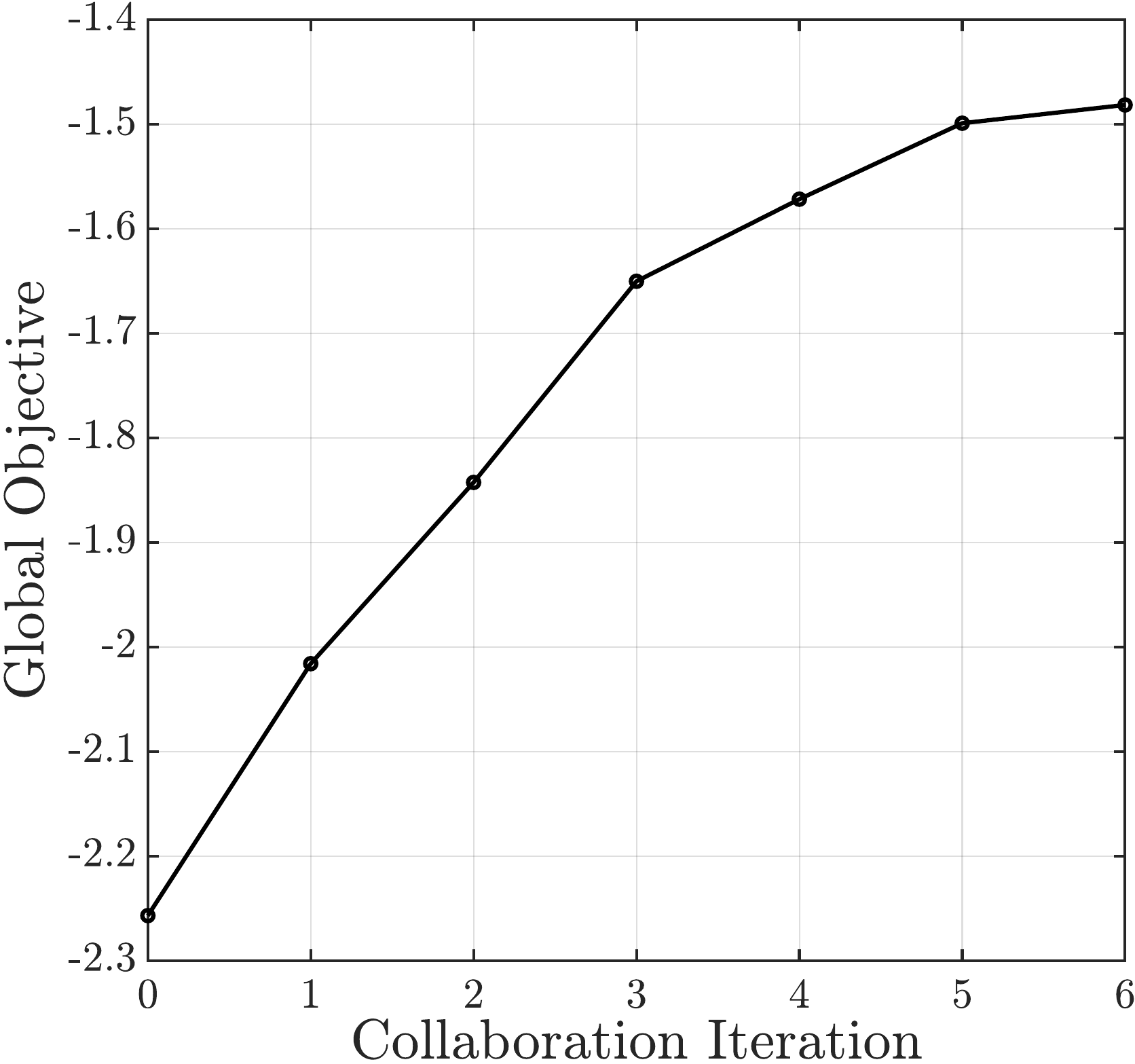}}
    \caption{Simulation of Algorithm~\ref{alg:collab_alg} with \( m = 4\) and \( N = 16 \) while having all teams have identical mission evaluation functions with the same density function \( \phi(q) \) in the locational cost in Equation~\ref{eq:lcost} for all teams, but different team weights. (a) The initial random allocation of robots to teams. (b) The final robot allocation. (c) The value of \( \mathcal{G} \) at each collaboration iteration until collaboration is no longer possible (hence, getting the allocation in (b)).}
    \vspace{-12pt}
    \label{fig:sim_same_phi}
\end{figure*}

\section{Conclusion} \label{sec:conclusion}

Optimizing collaboration within multi-robot systems is critical for enhancing operational efficiency, scalability, and robustness in complex environments. In this paper, we presented an ecology-inspired framework for facilitating collaboration in multi-team environments that achieves optimal performance for the entire system through altruistic behavior, where robots are considered as shareable system resources. However, many open directions remain for this work, including incorporating time-varying importance and mission evaluation functions, heterogeneous agents, and other mission applications such as search and rescue and wildfire management.

\bibliographystyle{IEEEtran}
\bibliography{mybib}

\end{document}